\documentclass[a4paper,12pt]{article}

\usepackage[english]{babel}
\usepackage[T1]{fontenc}
\usepackage[utf8]{inputenc}
\usepackage{amsmath}
\usepackage{amsfonts}
\usepackage{amsthm}

\usepackage{xspace}
\usepackage{tabularx}
\usepackage{indentfirst}
\usepackage{subfigure}
\usepackage[small]{caption}
\usepackage{eucal}
\usepackage{eso-pic}
\usepackage{url}
\usepackage{booktabs}
\usepackage{afterpage}
\usepackage{parskip}
\usepackage{listings}
\usepackage{fancyhdr}
\usepackage{textcomp}
\usepackage{cite}
\usepackage{multirow}
\usepackage{setspace}
\usepackage{enumerate}
\usepackage{tikz-cd}
\usepackage{mathtools}
\usepackage{graphicx}
\usepackage{geometry}
\usepackage{lipsum} 
\usepackage{fancyhdr} 
\usepackage{enumitem}
\usepackage{hyperref}
\usepackage{nameref}
\usepackage{comment}
\usepackage{algorithm}
\usepackage{algpseudocode}
\usepackage{bm}
\usepackage{listings}

\newtheorem{theorem}{Theorem}[section]
\newtheorem{lemma}[theorem]{Lemma}
\newtheorem{definition}[theorem]{Definition}

\newtheorem{proposition}[theorem]{Proposition}

\usepackage{authblk}

\newcommand{\N}{\mathbb{N}}

\newcommand{\F}{\mathbb{F}}

\newcommand{\B}{\mathbb{B}}
\newcommand{\K}{\mathbb{K}}
\newcommand{\G}{\mathcal{G}}

\newcommand{\WGT}{\operatorname{WGT}}
\newcommand{\WGP}{\operatorname{WGP}}
\newcommand{\Tr}{\operatorname{Tr}}

\newcommand{\AN}{\operatorname{AN}}
\newcommand{\AI}{\operatorname{AI}}
\newcommand{\supp}{\operatorname{Supp}}

\newcommand{\DLR}{\operatorname{DLR}}
\newcommand{\xalpha}{{\bm{x}^{\bm{\alpha}}}}
\newcommand{\xbeta}{\bm{x}^{\bm{\beta}}}

\title{An algebraic attack on stream ciphers with application to nonlinear filter generators and WG-PRNG}
\author{Carla Mascia$^1$, Enrico Piccione$^2$, and Massimiliano Sala$^1$} 
\affil{$^1$University of Trento, IT, $^2$University of Bergen, NO}
\date{}

\begin{document}
\maketitle

\begin{abstract}
In this paper, we propose a new algebraic attack on stream ciphers. Starting from a well-known attack due to Courtois and Meier, we design an attack especially effective against nonlinear filter generators. First, we test it on two toy stream ciphers and then we show that the level of security of WG-PRNG, one of the stream ciphers submitted to the NIST competition on Lightweight Cryptography, is less than that claimed until now. 
\end{abstract}

\section{Introduction}
Stream ciphers\cite{rueppel1986stream} are one of the main cryptographic primitives used in symmetric cryptography. Historically, the first stream ciphers were built with {\lq\lq linear\rq\rq} registers, where linearity is meant both in the register update function (which sends one state to the next) and in the output function, which computes the keystream as a function of the current state.
Purely linear registers are not used any longer because their state can be quickly recovered from a small portion of their produced keystream, e.g. by the Berlekamp-Massey algorithms \cite[Chapter 7]{berlekamp1968algebraic}. Since the use of linear structures translates into hardware implementations based on only a few XOR gates, which is highly desirable for practical applications, most modern stream ciphers retain some part of this original structure. Among the many competing stream designs, one has  recently attracted some interest: the so-called \textit{nonlinear filter generators} \cite{dichtl1997nonlinear}.  Indeed, they preserve a linear update for their state, composed of one or several linear registers, but they output their keystream via a nonlinear function of their state: this function is called \textit{filter}. The most notable example of these ciphers is the WG-PRNG, which was submitted to the NIST competition on Lightweight Cryptography \cite{nist}.

Traditionally, stream ciphers are attacked with two  approaches: correlation attacks, that exploit possible correlations between some part of the keystream and a portion of the initial state, and approximation attacks, where the nonlinear part is approximated by a linear component. The design defenses against these types of attacks rely on choosing nonlinear components with specific properties, such as high nonlinearity \cite{carlet2004new} and high correlation immunity\cite{siegenthaler1984correlation}. In recent years, a new family of attacks have emerged, the so-called
\textit{algebraic attacks}. Some interesting works in this direction are \cite{Armknecht2009, la2022stream,la2022algebraic}.
In this paper, we propose a new form of algebraic attack, which is especially effective against nonlinear filter generators. We show with two toy examples how the attack can be performed in practice. We also apply our attack to WG-PRNG and we provide a complexity estimate that shows a fatal weakness of this cipher. We also report previous attempts at breaking WG-PRNG with algebraic attacks and we discuss their shortcomings.\\

The paper is structured as follows: 
\begin{itemize}
    \item In Chapter 2, we collect all the notations, definitions and known facts needed in the remainder of the paper. We briefly illustrate the XL-Algorithm to solve Boolean equations systems and the algebraic attack to nonlinear filter generators presented in \cite{courtois2003algebraic}.
    \item In Chapter 3, we explain our improved algebraic attack in detail. 
    \item In Chapter 4, to validate our algebraic attack, first we apply it to two toy stream ciphers and then we show that it is feasible to perform it on WG-PRNG. We conclude showing that the security of WG-PRNG is less that claimed until now. For the sake of presentation, we will first describe the part regarding WG-PRNG, and then the one on the two toy ciphers. 
\end{itemize}

\section{Preliminaries}
In this section, we fix some notations and recall some known results. We denote by $\F_2$ the field with two elements, and by $R = \F_2[x_1,\dots ,x_n]$ the polynomial ring in $n$ variables over $\F_2$. Given a monomial $\xalpha=x_1^{\alpha_1} x_2^{\alpha_2} \cdots x_n^{\alpha_n}$ of $R$, the degree of $\xalpha$ is $\deg \xalpha = \sum_{i=1}^n \alpha_i$. For $d \in \N$, denote by $M_{\leq d}$ the set of all the square-free monomials in $R$ of degree at most $d$, that is \[
M_{\leq d} = \{ \xalpha \in R \ | \ \deg \xalpha \leq d \text{ and } \alpha_i \leq 1 \text { for } i=1, \dots, n\}.
\]

Let $\xalpha$ and $\xbeta$ be two distinct monomials of $R$. The \textit{degree reverse lexicographic order} (DLR) on $R$ is defined as ${\xalpha} \prec_{\DLR} {\xbeta}$ if either $\deg \xalpha < \deg \xbeta$ or $\deg \xalpha = \deg \xbeta$ and the rightmost nonzero component of the vector $\bm{\alpha} - \bm{\beta}$ is positive. Let $f \in R$, we denote by $\supp(f)$ the support of $f$, that is the set of all the non-zero terms of $f$.

For $n \in \N$, let $\B_n = \{f: \F_2^n \rightarrow \F_2\}$ denote the set of the Boolean functions in $n$ variables.  
Depending on the context, we represent $f \in \B_n$ as a square-free polynomial of $R$ or by means of its \textit{algebraic normal form} (ANF), namely as a polynomial of the quotient ring $R / L_n$, where $L_n = \langle x_1^2 - x_1, \dots, x_n^2 - x_n \rangle \subset R$ is the ideal generated by the field equations.

Let $f\in \B_n$, then
$$\AN(f)=\{g\in\B_n\,|\, fg=0\}$$
is called the \textit{set of annihilators} of $f$. Notice that
$\AN(f+ 1)=\{g\in\B_n\,|\, fg=g\}$. Moreover,
$\AN(f)$ and $\AN(f+ 1)$ are ideals of $R/L_n$, and $\AN(f)=\langle f+1 \rangle $ and $\AN(f+ 1)= \langle f \rangle.$


Let $f\in \B_n$, then
$$\AI(f)=\min\{\deg g \,|\, g\in \AN(f) \cup \AN(f+ 1),\, g\neq 0\}$$
is called the \textit{algebraic immunity} of $f$. By \cite{courtois2003algebraic}, it holds $\AI(f) \leq \left\lfloor\frac{n+1}{2}\right\rfloor$.

When $f = a+ \sum_{i=1}^n a_i x_i \in \B_n$, for some $a, a_i \in \F_2$, we say that $f$ is an \textit{affine Boolean function}. When $a=0$, $f$ is said to be \textit{linear}. 

A nonlinear filter generator is the combination of a linear feedback shift register (LFSR) and a nonlinear Boolean function. More precisely, 
\begin{definition}\label{nonlinearKSG}
A \emph{nonlinear filter generator} is a stream cipher that starts from an initial state $S=(s_0,\dots ,s_{n-1})\in\F_2^n$ and, at each clock $t\geq 0$, produces a keystream bit $z_t=F\left(L^t\left(S\right)\right)$, where
\begin{itemize}
\item $L\colon \F_2^n\to \F_2^n$ is the linear update function;
\item $F\in \B_n$ is the nonlinear output function.
\end{itemize}
$F$ is called \emph{filter function}.
\end{definition}

\subsection{Solving a Boolean equations system}\label{subsec: solving}

Let $f_1, \dots, f_t \in \B_n$ be nonlinear Boolean equations. Consider the following system
\begin{equation}\label{eq:system}
\begin{cases}
f_1(\bm{x})=0,\\
f_2(\bm{x})=0,\\
\quad \vdots\\
f_t(\bm{x})=0.
\end{cases}
\end{equation}

Several methods have been developed for solving such a system, such as the XL-Algorithm \cite{courtois2000efficient} and Gr\"obner basis techniques \cite{faugere1999new,faugere2002new}.

\subsection*{XL-Algorithm}
The XL-Algorithm, introduced by Courtois, Klimov, Patarin, and Shamir in \cite{courtois2000efficient}, is a computation method for solving a system as the one in \eqref{eq:system}. Assume that all the $f_i$, for $i =1, \dots, t$, have the same positive degree $d\in\N$. Fix $D \in \N$ such that $D \geq d$. The idea of the XL-Algorithm is to generate all the possible equations of at most degree $D$ that satisfy the system. The algorithm performs the following four steps.
\begin{enumerate}
\item \textbf{Multiply}: For each $i=1,\dots ,t$, generate the equations 
$$\{ \xalpha f_i(\bm{x})=0 \ | \  \xalpha \in M_{\leq D-d} \}.$$
\item \textbf{Linearize}: Consider each monomial in $M_{\leq D}$ as an independent variable and perform Gaussian elimination on the equations obtained in Step 1. The ordering on the monomials must be such that all the terms containing one (fixed) variable (say $x_1$) are eliminated last.

\item \textbf{Solve:} If Step 2 yields at least one univariate equation in the powers of $x_1$, solve this equation over $\F_2$. If not, algorithm fails.
\item \textbf{Repeat:} Simplify the equations and repeat the process to find the values of the other variables.
\end{enumerate}
Other versions of the XL-algorithm can be found in \cite{courtois2002higher}. If $t$ is big enough, we expect to find one solution for the system. In this case, the complexity of XL will be essentially the complexity of one single Gaussian reduction in Step 2.\\
Let $N$ be the number of equations generated in XL, and $T$ be the number of monomials in $M_D$. Then 
\begin{equation}\label{eq:estimates}
N=t \cdot\left(\sum_{i=0}^{D-d}\left(\begin{array}{l}
n \\
i
\end{array}\right)\right) \approx t \cdot\left(\begin{array}{c}
n \\
D-d
\end{array}\right) \quad \text{and} \quad T=\sum_{i=0}^{D}\left(\begin{array}{l}
n \\
i
\end{array}\right) \approx\left(\begin{array}{l}
n \\
D
\end{array}\right).
\end{equation}
Let $\mathcal{I}$ be the number of linearly independent equations in XL. Clearly, $\mathcal{I} \leq T$. In practice, $\mathcal{I} \leq N$. The main heuristic behind XL is that for some $t, D$, we have always $N \geq T$. Then we expect that if $\mathcal{I} \geq T-D$, it is possible, by Gaussian elimination, to obtain one equation in only one variable, and XL will work. Otherwise, we need a bigger $t$ or $D$. 

In conclusion, if the algorithm works, the time complexity is approximately $\mathcal{O}\left(T^\omega\right)$, where $2<\omega <3$ depends on the implementation of the matrix multiplication algorithm. The minimum known value for $\omega$ is $\omega\approx 2.3728596$ (see \cite{alman2021refined}). But in real implementations, $\omega=\log_2(7)\approx 2.8073549$ (see \cite{strassen1969gaussian}), and in our estimates we will consider this last approximation. 

\subsection{Algebraic Attack}
An algebraic attack to a nonlinear filter generator with nonlinear output function $F$, if we know $t$ keystream bits, consists in solving a nonlinear Boolean equations system, namely
\begin{equation}\label{eq: system with F}
    F\left(L^{i}\left(\bm{x}\right)\right) =  z_{i}, \qquad \text{ for } i=0, \dots, t-1,
\end{equation}

in order to recover the initial state. 

In the following, we will report the generic algebraic attack designed by Courtois and Meier, and we refer to the original work \cite{courtois2003algebraic} for more detail.\\
The main idea behind this attack is to decrease the degree of the original system by multiplying each equation in \eqref{eq: system with F}, that are usually of high degree, by a well chosen $g\in \B_n$. The resulting equations are
\begin{equation}\label{algEqCM}
F\left(L^i\left(\bm{x}\right)\right)g\left(L^{i}\left(\bm{x}\right)\right)=z_i g\left(L^i\left(\bm{x}\right)\right), \qquad i=0, \dots, t-1,
\end{equation}
which are of substantially lower degree.
Then, if $z_i = 0$ (resp. $z_i = 1$), we can choose $g \in \AN(F+1)$ (resp. $g \in \AN(F)$), and we get an equation of low degree on the initial state bits, that is $g({\bm x}) = 0$. The smaller the algebraic immunity of $F$ is, the lower degree the resulting equation has. If we get one such equation for each of sufficiently many keystream bits, we obtain a very overdefined system of multivariate equations that can be solved efficiently.

\section{An improved algebraic attack}\label{sec: our attack}

The core idea of our new algebraic attack is to use many annihilators simultaneously, instead of only one, and provide a good estimate of the number of keystream bits needed to perform the attack, which is strictly related to the number of linearly independent equations after the multiply phase in the XL-Algorithm. Indeed, by increasing the number of linearly independent equations, we need fewer keystream bits than the ones required in the Courtois and Meier's attack.

Before presenting our attack, we need some preliminary results.

\begin{lemma}\label{lemma: g divides sum} 
Let $R= \mathbb{K}[x_1, \dots, x_n]$ be a polynomial ring over any field $\K$ and let $A = \{f_1, \dots, f_k\} \subset R^{\prime}= \mathbb{K}[x_1, \dots, x_m]$, with $m < n$, be a set of linearly independent polynomials. Let $g, g_1, \dots, g_k \in R^{\prime \prime} = \mathbb{K}[x_{m+1},\dots, x_n]$, with $g \not \in \mathbb{K}$. If $h = \sum_{i=1}^k g_if_i$ is equal to $gf$, for some $f \in R^{\prime}$, then either $h=0$ or, for all $i=1, \dots, k$, $g$ divides $g_i$. 
\end{lemma}

\begin{proof}
We prove the statement by induction on the degree of $h = \sum_{i=1}^k g_if_i$. 

Let $\deg h = 0$. Since $g$ is not a constant polynomial, therefore $\deg g >0$. By hypothesis, $g$ divides $h$, but $\deg g > \deg h$. It follows that $f = h=0$.

Let $\deg h > 0$. For any $i=1, \dots, k$, we write $g_i = gh_i + r_i$, with $\deg r_i < \deg g$ and $h_i, r_i \in R^{\prime \prime}$. Note that at least one $h_i$ has to be nonzero, otherwise all the $g_i$ have degree less than $\deg g$, and then $g$ can not divide $h$. We have 
\[
gf = \sum_{i=i}^k (gh_i + r_i)f_i = g \sum_{i=1}^k h_if_i + \sum_{i=1}^k r_if_i.
\]
Therefore, $g$ divides $h^{\prime} = \sum_{i=1}^k r_if_i$. As $\deg h^{\prime} < \deg h$, by induction hypothesis, either $h^{\prime} = 0$ or, for all $i=1, \dots, k$, $g$ divides $r_i$. The latter can not happen, as $\deg r_i < \deg g$, for all $i=1, \dots, k$. Then $h^{\prime} = 0$. Since $\{f_1, \dots, f_k\} \subset R^{\prime}$ is a set of linearly independent polynomials and $r_i \in R^{\prime \prime}$, it holds that $r_i = 0$, for all $i=1, \dots, k$. Therefore, $g_i = gh_i$, for all $i=1, \dots, k$ and the statement is proved. 
\end{proof}


\begin{proposition}\label{prop: lin ind}
Let $R= \mathbb{K}[x_1, \dots, x_n]$ be a polynomial ring over any field $\K$ and let $A = \{f_1, \dots, f_r\} \subset R$ be such that, for any $i=1, \dots, r$, $f_i \in R'= \mathbb{K}[x_1, \dots, x_m]$, with $m \leq n$. 
Let $B=\{m_{1,1}f_1, \dots, m_{1,k_1} f_1, \dots,  m_{r,1}f_r, \dots, m_{r,k_r} f_r\} \subset R$, where for $i=1, \dots, r$ and $j=1, \dots, k_i$, $m_{i,j}$ are nonzero monomials in $R''=\mathbb{K}[x_{m+1}, \dots, x_{n}]$, $m_{i,j} \neq cm_{i,k}$ for all $j \neq k$ and $c \in \mathbb{K}$.
If $A$ is a set of linearly independent polynomials in $R'$, then $B$ is a set of linearly independent polynomials in $R$.
\end{proposition}

\begin{proof}
We prove the statement by induction on the number $r$ of polynomials in $A$. 

Let $r=2$. Suppose there exist $c_i,d_j \in \mathbb{K}$, not all zero, such that
\[
\sum_{i=1}^{k_1} c_i m_{1,i} f_1 + \sum_{j=1}^{k_2} d_j m_{2,j} f_2=0.
\]
Let 
\[
g_1 = \sum_{j=1}^{k_1} c_j m_{1,j}, \qquad \text{ and } \qquad g_2= - \sum_{j=1}^{k_2} d_j m_{2,j}.
\]
Therefore, we have $f_1g_1 = f_2g_2$. Since $R'$ and $R''$ are polynomial rings in two disjoint sets of variables, $g_1, g_2\in R''$, whereas $f_1, f_2 \in R'$, it follows that $g_1 = c g_2$ and $f_1= d f_2$, for some $c,d \in \mathbb{K}$, that is a contradiction, as $f_1, f_2$ are supposed linearly independent. 

Let $r>2$, and suppose there exist $c_{i,j} \in \mathbb{K}$, for $i=1, \dots, r$ and $j=1, \dots, k_i$, not all zero, such that
\[
\sum_{i=1}^r \sum_{j=1}^{k_i} c_{i,j} m_{i,j} f_i = 0.
\]
Let 
\[
g_1 = \sum_{j=1}^{k_1} c_{1,j} m_{1,j}, \qquad \text{ and}  \qquad g_i= -  \sum_{j=1}^{k_i} c_{i,j} m_{i,j}, \; \text{ for } i=2, \dots, r.
\]
We may assume, without loss of generality, $c_{1,1} \neq 0$, and then  $g_1 \neq 0$, since $m_{1,j} \neq c m_{1,k}$ for each $j \neq k$ and $c \in \mathbb{K}$, by hypothesis. Therefore,
\[
g_1 f_1 = g_2f_2 + \cdots +g_r f_r.
\]
Let $S = \{s \ | \ g_s = c_s g_1, \text{ for some } c_s \in \mathbb{K}\} \subseteq \{2, \dots,r\}$ and $T = \{2, \dots, r\} \setminus S$. Then,
\begin{equation}\label{eq: g1f1 = }
        g_1 \left(f_1 - \sum_{s \in S}  c_s f_s\right) = \sum_{t \in T} g_t f_t. 
\end{equation}
$g_1$ divides the term on the left in Equation $\eqref{eq: g1f1 = }$, then it must divide the one on the right, as well. By Lemma \ref{lemma: g divides sum}, either $\sum_{t \in T} g_t f_t=0$ or, for all $t \in T$, $g$ divides $g_t$. Assume the latter is true, then 
\[
g_1 \left(f_1 - \sum_{s \in S}  c_s f_s\right) = g_1\left(\sum_{t \in T} \frac{g_t}{g_1} f_t\right). 
\]
It should hold $\sum_{t \in T} \frac{g_t}{g_1} f_t \in R^{\prime}$, then $\frac{g_t}{g_1} \in \mathbb{K}$, that is $g_t = cg_1$, for some $c \in \mathbb{K}$, but this is a contradiction since $t \not \in S$. It follows that $\sum_{t \in T} g_t f_t=0$. The set $T$ has cardinality smaller than $r$, then, by applying the induction hypothesis, $c_{b,j} = 0$, for all $b \in B$ and $j=1, \dots, k_b$. Hence, we obtain
\[
g_1\left(f_1 - \sum_{s \in S} c_s f_s\right) = 0.
\]
As $g_1 \neq 0$, $f_1 - \sum_{s \in S} c_s f_s$ should be zero. But this is a contradiction for the linearly independence of the $f_i$'s. 
\end{proof}

Consider a nonlinear filter generator with an $n$-bit inner state. Denote by $F \in \B_n$ its nonlinear output function. It takes in input $m \leq n$ bits of the inner state. Therefore, up to rename the variables, we may consider $F$ as a square-free polynomial of $\F_2[x_1, \dots, x_m]$.

Our algebraic attack consists of the following steps:
\begin{enumerate}
    \item Define two ideals of $\F_2[x_1, \dots, x_m]$, $I_0$ and $I_1$, as 
    \begin{equation*}
    I_0=\langle F \rangle +  L_m   \qquad \text{ and } \qquad I_1=\langle F +1\rangle +  L_m,
    \end{equation*}
    where $L_m \subset \F_2[x_1,\dots ,x_m]$ is the ideal generated by the field equations. Compute the reduced Gr\"obner bases, $\G_0$ and $\G_1$, with respect to $\prec_{\DLR}$, of $I_0$ and $I_1$, respectively. 
 \item Select from $\mathcal{G}_0$ (resp.  $\mathcal{G}_1$) the maximal set $\mathcal{G}_0^{\prime}$ (resp.  $\mathcal{G}_1^{\prime}$) of square-free polynomials  such that the degree of all the polynomials is not greater than $\deg F$ and $\mathcal{G}_0^{\prime}$ (resp.  $\mathcal{G}_1^{\prime}$) generates $I_0$ (resp. $I_1$). 
    \item Denote by $d = \max\{\deg f \ | \ f \in \G_0^{\prime} \cup \G_1^{\prime}\}$. Fix $D \geq d$. Multiply each $f$ in $\G_0^{\prime}$ (resp. $\G_1^{\prime}$) by all the square-free monomials $\bm{x}^{\bm{\alpha}} \in \F_2[x_1, \dots, x_m]$ such that $\deg (\bm{x}^{\bm{\alpha}} f) \leq D$. Reduce all the polynomials by $L_m$ and denote by $\mathcal{S}_0 \subset \B_m$ (resp. $\mathcal{S}_1$) the resulting set of distinct polynomials. 
    \item  Select from $\mathcal{S}_0$ (resp.  $\mathcal{S}_1$) the maximal set of linearly independent polynomials. Denote it by $\mathcal{S}^{\prime}_0$ (resp. $\mathcal{S}^{\prime}_1$). 
    \item Compute the number of required keystream bits as 
    \begin{equation}\label{eq: t}
        t = \left \lceil{ \frac{\sum_{i=0}^D\binom{n}{i}}{\min \{k_0^{\prime}, k_1^{\prime}\}}}\right \rceil , \quad \text{ where } k_i^{\prime} = \sum_{f \in \mathcal{S}^{\prime}_i} \sum_{i=0}^{D - \deg f} \binom{n-m}{i}, \; \text{ for } i=0,1.
    \end{equation}
    If $t$ is greater than the maximum number of consecutive keystream bits fixed for the stream cipher, then the attack is infeasible.
    \item Solve, by using the XL-Algorithm with the fixed $D$, the system
    \begin{equation}\label{eq: system with GB}
        \begin{cases}
        f_{0,z_i}\left(L^i\left(\bm{x}\right)\right)=0,\\
        f_{1,z_i}\left(L^i\left(\bm{x}\right)\right)=0,\\
        \quad \vdots & \qquad \text{ for } i=0, \dots, t-1,\\
        \quad \vdots \\
        f_{k_{z_i},z_i}\left(L^i\left(\bm{x}\right)\right)=0.
        \end{cases}
        \end{equation}
    where, for $j=0,1$, the polynomials $f_{0,j}, \dots, f_{k_j,j}$ are all the square-free polynomials in  $\G_j^{\prime}$.\\
\end{enumerate}

First of all, note that in $\B_n$ all the polynomials of $I_0$ (resp. $I_1$) are  annihilators  of  $F+1$ (resp. $F$), and then $\G_0^{\prime} \subseteq \AN(F+1)$ (resp. $\G_1^{\prime} \subseteq \AN(F)$). Moreover, by Step 2., all the polynomials in $\G_0^{\prime} \cup \G_1^{\prime}$ have degree at most $\deg F$. Therefore, the system \eqref{eq: system with GB} is obtained from system \eqref{eq: system with F} by multiplying each equation for more than one annihilator, and we have decreased the degree of the equations in the system.

The delicate part is to determine $t$, that is the number of needed keystream bits to solve the system and hence to perform the attack. In fact, the XL-Algorithm successfully finishes if there are enough linearly independent equations after multiply phase. In our situation, it is not simple to estimate the linear dependencies that arise from the linear update function $L$. To be more precise, there could exist  $t_1,\dots ,t_\ell \in \N$ and monomials $m_1,\dots ,m_{\ell} \in \B_n$ of degree at most $D-d$ such that  $$\{m_1f_{j_1,z_{t_1}}(L^{t_1}\left(\bm{x}\right)),\dots ,m_{\ell}f_{j_\ell,z_{t_\ell}} (L^{t_\ell}\left(\bm{x}\right))\}$$ is not a linearly independent set. 

If we suppose that after the multiply phase all the equations are linearly independent, except for a negligible part of them, then we bump into an underestimation of the needed keystream bits, and then of the security of the nonlinear filter generator. Steps from 2. to 5. are aimed at providing a fair $t$. Up to Step 4., we compute a maximal set of square-free polynomials in $\F_2[x_1, \dots, x_m]$ that are linearly independent. In Step 5, we exploit Proposition \ref{prop: lin ind} to estimate the number of linearly independent polynomials in $\F_2[x_1, \dots, x_n]$, and then, the value $t$. In detail, at each clock $t$, every polynomial $f \in \mathcal{S}^{\prime}_i$ is multiplied by $\sum_{i=0}^{D - \deg f} \binom{n-m}{i}$ monomials. Since the polynomials in $\mathcal{S}^{\prime}_i$ are linearly independent, by Proposition \ref{prop: lin ind} the number of linearly independent equations,  at each clock $t$, is given by 
\[
k_i^{\prime} = \sum_{f \in \mathcal{S}^{\prime}_i} \sum_{i=0}^{D - \deg f} \binom{n-m}{i},
\]
where either $i=0$ or $i=1$, depending on the value of $z_t$. We can approximate the number of linearly independent equations, which we obtain after the multiply phase of XL, with $t \cdot \min \{k_0^{\prime}, k_1^{\prime}\}$. To guarantee that the system can be solved by Gaussian elimination, we impose the condition 
\[
t \cdot \min \{k_0^{\prime}, k_1^{\prime}\} \geq \sum_{i=0}^D\binom{n}{i},
\]
and we get the estimation for $t$.

\section{Applications of our attack}
In this section, first we describe the stream cipher WG-PRNG, we show how to apply it on WG-PRNG. Moreover, we effectively perform our algebraic attack on two toy stream ciphers.

\begin{subsection}{Testing our attack on WG-PRNG}

\subsubsection{Specifications of WG-PRNG}\label{sec: specifications}

WG-PRNG, which was submitted to the NIST competition on Lightweight Cryptography \cite{nist}, is a nonlinear filter generator that operates over the finite field $\mathbb{F}_{2^{7}}$, defined using the primitive polynomial $f(y)=y^{7}+y^{3}+y^{2}+y+1 \in \mathbb{F}_2[y]$. Let $\omega$ be a root of $f(y)$.  By using the polynomial basis $\{1, \omega, \dots, \omega^6\}$, any $x \in \F_{2^7}$ can be written as $x = \sum_{i=0}^6 x_i \omega^i$, for $x_i \in \F_2$.\\
The function $\WGP\colon \F_{2^7}\to \F_{2^7}$ defined as 
$$\WGP\left(x\right)= \left(x+1\right) +\left(x+1\right)^{33}+\left(x+1\right)^{39}+\left(x+1\right)^{41}+\left(x+1\right)^{104} +1$$
is called the \textit{WG permutation} over $\F_{2^7}$. The function $\WGT\colon \F_{2^7}\to \F_{2}$ defined as
$$\WGT\left(x\right)=\Tr\left(\WGP\left(x\right)\right)$$
is called the \textit{WG transformation} over $\F_{2^7}$, where $\Tr : \F_{2^7} \rightarrow \F_2$ denotes the trace function defined by $\Tr(x) = x_0+x_5$.
A decimated WG permutation and decimated WG transformation over $\F_{2^7}$ are defined as $\WGP(x^d)$ and
$\WGT(x^d)$, respectively, with $\gcd(d, 2^7 -1) = 1$. The filter function of WG-PRNG is the decimated WG transformation with $d=13$. We have computed its algebraic normal form and it is given by
\begin{align}\label{eq: algebraic normal form}
\begin{split}
\WGT=& \ x_2x_3x_4x_5x_6x_7+x_1x_2x_3x_4x_6+x_1x_2x_3x_5x_6+x_1x_2x_4x_5x_6\\
&+x_2x_3x_4x_5x_6+x_1x_2x_3x_5x_7+x_1x_3x_4x_5x_7+x_2x_3x_4x_5x_7\\
&+x_1x_2x_3x_6x_7+x_1x_3x_4x_6x_7+x_2x_3x_4x_6x_7+x_1x_4x_5x_6x_7\\
&+x_2x_4x_5x_6x_7+x_1x_2x_3x_5+x_1x_2x_3x_6+x_2x_3x_4x_6+x_1x_2x_5x_6\\
&+x_2x_4x_5x_6+x_3x_4x_5x_6+x_2x_3x_4x_7+x_1x_2x_5x_7+x_2x_3x_5x_7\\
&+x_1x_4x_5x_7+x_2x_4x_5x_7+x_2x_3x_6x_7+x_1x_4x_6x_7+x_2x_5x_6x_7\\
&+x_3x_5x_6x_7+x_4x_5x_6x_7+x_1x_2x_3+x_1x_2x_5+x_1x_3x_5+x_2x_3x_5\\
&+x_1x_2x_6+x_1x_4x_6+x_2x_4x_6+x_3x_4x_6+x_4x_5x_6+x_1x_2x_7\\
&+x_1x_4x_7+x_3x_4x_7+x_4x_5x_7+x_1x_6x_7+x_3x_6x_7+x_5x_6x_7+x_3x_4\\
&+x_4x_5+x_1x_6+x_4x_6+x_2x_7+x_4x_7+x_5x_7+x_1+x_4+x_6+x_7.
\end{split}
\end{align}

Henceforth, we will write $\WGT$ and $\WGP$, even if we refer to their decimated versions with $d=13$. 
The inner state of WG-PRNG consists of $37$ words $S_i$, each of $7$ bits, for a total of $259$ bits. The WG-PRNG has two phases: an initialization phase and a running phase. The output is produced only in the running phase.

\textbf{Initialization phase: }Let $\left(S_{36},S_{35},\dots ,S_1,S_0\right)$ denote the initial state.  A random seed is loaded into the internal state and then the state update function is applied 74 times. For $t=0, \dots, 73$, the state update function is given by 
\begin{align*}
    S_{37+t} = &\WGP\left( S_{36+t}\right)+S_{31+t} + S_{30+t} + S_{26+t} + S_{24+t} + S_{19+t}  \\
    & + S_{13+t} + S_{12+t} + S_{8+t} + S_{6+t} +\left(\omega \cdot S_{t}\right).
\end{align*}

\textbf{Running phase: }In this phase, the inner state is updated according to the following LFSR feedback function:
\begin{align}\label{eq:update}
\begin{split}
    S_{37+t} = & \ S_{31+t} + S_{30+t} + S_{26+t} + S_{24+t} + S_{19+t}  \\
    & + S_{13+t} + S_{12+t} + S_{8+t} + S_{6+t} +\left(\omega \cdot S_{t}\right).
    \end{split}
\end{align}

At each clock cycle $t \geq 0$, a pseudorandom bit is produced by applying WGT on the last word of the register. A pseudorandom bit sequence $(z_{t})_{t\in\N}$ is produced by WG-PRNG as
$$z_{t}=\WGT\left(S_{110+t}\right).$$

\subsubsection{The previous algebraic attack on WG-PRNG}

We observe that the Courtois and Meier's attack is not feasible. 

Let $\bm{x}=\left(x_1,x_2,\dots ,x_{259}\right)$ be the inner state of WG-PRNG after the initialization phase. Denote by $L$ the update function of WG-PRNG that maps a state into the next one (see Equation \eqref{eq:update}). Let $\WGT$ be the filter function of WG-PRNG and consider the natural extension of $\WGT$ in $\B_{259}$. Suppose we retrieved $t$ keystream bits. Without loss of generality, we may assume that they are the first $t$ keystream bits, namely $z_0,\dots ,z_{t-1}$. Therefore, we have the following equations:
\begin{equation}\label{eq: system WGT}
\WGT \left(L^i\left(\bm{x}\right)\right)=z_i, \quad \text{ for } i=0, \dots, t-1.
\end{equation}

Hence, we get the system \eqref{eq:system}, with $f_i({\bm x}) = \WGT \left(L^{i-1}\left(\bm{x}\right)\right) +z_{i-1} $, for $i=1, \dots, t$.

 According to the algebraic attack of Courtois and Meier, to decrease the degree of the system equations, we can multiply any equation in \eqref{eq: system WGT} by an annihilator of $\WGT$ and the one of $\WGT +1$, depending on whether the keystream bit is 1 or 0. Since $\AI(\WGT) = \AI(\WGT +1) = 3$, let $f \in \AN(\WGT +1)$ and $g\in \AN(\WGT)$ be such that $\deg f = \deg g = 3$. Then, for $i=0,\dots ,t-1$, if $z_i=0$, we get $f\left(L^i\left(\bm{x}\right)\right)=0,$
 otherwise $g \left(L^i\left(\bm{x}\right)\right)=0.$
As reported in \cite{courtois2002higher}, the XL-Algorithm works if $$t\geq \binom{259}{3}  \approx 2^{21.45}$$ 

To avoid this attack, the designers conveniently restricted the number of consecutive output bits up to $2^{18}$.

\subsubsection{Applying our attack on WG-PRNG}
The main aim of our work is not to perform effectively the attack described in Section \ref{sec: our attack} on WG-PRNG, but to estimate how many keystream bits one needs to perform successfully the attack on WG-PRNG. We will show that knowing less than $2^{18}$ keystream bits, it is possible to recover the initial state, that is the security of the WG-PRNG is less than the one stated by the designer. 

We will describe in detail the steps from 1. to 5. in order to compute $t$.

\begin{enumerate}
    \item We first consider $\WGT$ as a square-free polynomial in $\F_2[x_1,\dots ,x_7]$, indeed its algebraic normal form involves only the variables $x_1,\dots ,x_7$ (see Equation \eqref{eq: algebraic normal form}). We set the two ideals of $\F_2[x_1,\dots ,x_7]$, $I_0$ and $I_1$, as
    \begin{equation*}
    I_0=\langle\WGT\rangle +  L_7   \qquad \text{ and } \qquad I_1=\langle\WGT +1\rangle +  L_7,
    \end{equation*}
    where $L_7 \subset \F_2[x_1,\dots ,x_7]$ is the ideal generated by the field equations. The reduced Gr\"obner bases with respect to $\prec_{\DLR}$ of $I_0$ and $I_1$ are 
    \begin{equation*}
    \G_0=\{f_0, f_1, f_2, \dots,f_{30}\} \cup \{x_i^2 -x_i \ \mid \ i=1, \dots, 7\}
    \end{equation*}
    and 
    \begin{equation*}
    \G_1=\{g_0, g_1, g_2, \dots,g_{30}\} \cup \{x_i^2 -x_i \ \mid \ i=1, \dots, 7\},
    \end{equation*}
    respectively, where $\deg f_0= \deg g_0 = 3$ and $\deg f_j= \deg g_j = 4$, for all $j=1,\dots ,30$.
    \item We set $\G_0^\prime =\{f_0, f_1, f_2, \dots,f_{30}\}$ and $\G_1^\prime= \{g_0, g_1, g_2, \dots,g_{30}\}$, since all the polynomials $f_i, g_j$ are square-free and of degree 3 or 4, that is smaller than $\deg F = 6$. 
    \item Let $d=4$, and we fix $D \in \{4, 5, 6, 7\}$. We multiply each $f$ in $\G_0^{\prime}$ (resp. $\G_1^{\prime}$) by all the square-free monomials $\bm{x}^{\bm{\alpha}} \in \F_2[x_1, \dots, x_7]$ such that $\deg (\bm{x}^{\bm{\alpha}} f) \leq D$. After reducing all the polynomials by $L_7$, we get the sets $\mathcal{S}_0, \mathcal{S}_1 \subset \B_7$.
    \item We select from $\mathcal{S}_0$ (resp.  $\mathcal{S}_1$) the maximal set of linearly independent polynomials, and we obtain $\mathcal{S}_0^\prime$ (resp.  $\mathcal{S}_1^\prime$). Both $\mathcal{S}_0^\prime$ and $\mathcal{S}_1^\prime$ consists of 64 polynomials. In particular, let $\left(\mathcal{S}^{\prime}_i\right)_r=\{f\in\mathcal{S}^{\prime}_i\,|\, \deg(f)=r\}$, for $i=0,1$. Then,
$$|\left(\mathcal{S}^{\prime}_i\right)_r|=\begin{cases}0 &\text{if } r< 2,\\
1 &\text{if }r=3, \\ 
34 &\text{if }r= 4,\\
21 &\text{if }r= 5,\\
7 &\text{if }r=6,\\
1 &\text{if }r=7.\\
\end{cases}$$
\item In this case, we have $k_0^\prime =k_1^\prime$, and, by varying $D$, they are equal to:
$$k_i^\prime=\begin{cases}
287\approx 2^{8.16} &\text{if }D=4,\\
40502\approx 2^{15.31} &\text{if }D=5,\\
3756585\approx 2^{21.85} &\text{if }D=6,\\
258089371\approx 2^{27.94} &\text{if }D=7.
\end{cases}$$
By means of the Equation \eqref{eq: t}, we compute the number of required keystream bits, as summarize in Table \ref{table: tabWG}, together with the time complexity. The latter is computed as $\mathcal{O}\left(T^{\omega}\right)$, as reported in Subsection \ref{subsec: solving}.
    
\end{enumerate}
 
 {\renewcommand\arraystretch{1.3}
\begin{center}
\begin{table}[H]
\centering
\begin{tabular}{|l|l|l|}
\hline
\textbf{D} & \textbf{Keystream Bits Collected ($t$)} &\textbf{Time Complexity} \\ \hline
$4$                         & $2^{19.31}$                               & $2^{77.06}$                               \\ \hline
$5$                         & $2^{17.84}$                               & $2^{92.98}$                               \\ \hline
$6$                         & $2^{16.72}$                               & $2^{108.15}$                              \\ \hline
$7$                         & $2^{15.80}$                               & $2^{122.68}$                              \\ \hline
\end{tabular}
\caption{}
\label{table: tabWG}
\end{table}
\end{center}}

As shown in Table \ref{table: tabWG}, the attack is not feasible if $D=4$, as it needs more than $2^{18}$ keystream bits. However, for both $D=5$ and $D=6$, it is possible to carry out the attack knowing $2^{17.84}$ and $2^{16.72}$, respectively, and so the security level is less than 128 bits, contradicting the claim in \cite{altawy2020wage}. Finally, for $D \geq 7$, the time complexity is worse than brute force.



\end{subsection}

\begin{subsection}{Testing our attack on two toy stream ciphers}
To validate our algebraic attack, we will define a general construction of a scaled version of WG-PRNG, and we will test our attack on two instances.

We consider nonlinear filter generators that operate over the finite field $\mathbb{F}_{2^{7}}$, defined using the primitive polynomial $f(y)=y^{7}+y^{3}+y^{2}+y+1 \in \mathbb{F}_2[y]$, as in WG-PRNG. 
Moreover, at each clock $t\geq0$, it produces a keystream $z_t=F(L^t(S))$, where $S$ is the initial state of length $n=7a$ (for WG-PRNG, $a=37$), $F$ is the nonlinear output function of WG-PRNG, i.e. WGT, and $L$ is a linear update function satisfying the following properties:
\begin{itemize}
    \item it is a primitive polynomial,
    \item it has an odd number of terms,
    \item the constant term is $\omega$, where $\omega$ is a root of $f(y)$,
    \item the coefficients of the terms of degree nonzero are in $\F_2$.
\end{itemize}

We have tested our algebraic attack on two nonlinear filter generators, as defined above. We set $a=3$ and $a=5$, respectively, and the linear update functions are $L_1 = x^3 +x+\omega$ and $L_2 = x^5 +x^2+\omega$, respectively.

We have studied only the case $D=5$.

Note that, since the filter function is the same of WG-PRNG, the first four steps of the attack returns the same results of the first four steps performed on WG-PRNG.

First we consider the toy that has $L_1$ as linear update function. We compute $k^\prime_0$ and $k^\prime_1$. For $i=0,1$, we have $k_i^\prime=637$. According to our estimates, we need to collect $t=44$ keystream bits to perform an attack. We run computations to determine how much $tk^\prime$ is close to the number of linearly independent equations obtained at the end of the algorithm. The number of linearly independent equations is $26544$, while $tk_i^\prime =28028$.

Now we consider the toy that has $L_2$ as linear update function. We compute $k^\prime_0$ and $k^\prime_1$. For $i=0,1$, $k_i^\prime=1414$. According to our estimates, we need to collect $t=272$ keystream bits to perform an attack. We run computations to determine how much $tk^\prime$ is close to the number of linearly independent equations obtained at the end of the algorithm. The number of linearly independent equations is $353199$, while $tk_i^\prime =384608$.

The relevant result of these two experiments is that in both cases, by knowing only the number of keystream bits we supposed was sufficient, we completely recover the initial state (which was chosen randomly). 

\end{subsection}

\textbf{Acknowledgement.} The computation of this work has been obtained thanks to Magma \cite{bosma1997magma} and to the server of the Laboratory of Cryptography of the Department of Mathematics, University of Trento. The results in this paper appear partially in the MSc thesis of the second author, who thanks his supervisors (the other two authors). This work has been partially presented in Cryptography and Coding Theory Conference
2021, organized by the group UMI (Italian Mathematical Union) ``Crittografia e Codici'' and by De Componendis Cifris, in September 2021. 

\bibliographystyle{plain}
\bibliography{Bibliography.bib}
\thispagestyle{empty}

\vspace{0.2cm}

\small{\textit{Email adressess}: \{carla.mascia, massimiliano.sala\}@unitn.it, enrico.piccione@uib.no.}
\end{document}